\documentclass[runningheads,letterpaper]{llncs}

\usepackage{amssymb}
\usepackage{amsmath}
\usepackage{graphicx}
\usepackage{xspace}
\usepackage{multicol}
\usepackage{url}
\usepackage{cbneed-tfp2010}

\begin{document}
\mainmatter

\title{Evaluating Call-By-Need on the Control Stack}

\author{Stephen Chang,\thanks{Partially supported by grants from the National Science Foundation.} David {Van Horn},\thanks{Supported by NSF Grant 0937060 to the CRA for the CIFellow Project.} and Matthias Felleisen$^{\star}$}
\institute{PLT \& PRL, Northeastern University, Boston, MA 02115}

\maketitle

\markboth{Stephen Chang, David Van Horn, and Matthias Felleisen}
         {Evaluating Call-By-Need on the Control Stack}

\begin{abstract}
 Ariola and Felleisen's call-by-need $\lambda$-calculus replaces a variable
 occurrence with its value at the last possible moment. To support this
 gradual notion of substitution, function applications---once
 established---are never discharged. In this paper we show how to translate
 this notion of reduction into an abstract machine that resolves variable
 references via the control stack. In particular, the machine uses the {\em
 static address\/} of a variable occurrence to extract its current value from
 the {\em dynamic control stack\/}. 
\end{abstract}


\section{Implementing Call-by-need}

Following Plotkin~\cite{Plotkin75callbyname}, Ariola and Felleisen
characterize the by-need $\lambda$-calculus as a variant of $\beta$:
\[
(\lambda x.E[x]) ~ V = (\lambda x.E[V]) ~ V \enspace\text,
\] 
and prove that a machine is an algorithm that searches for a
(generalized) value via the leftmost-outermost application of this new 
reduction~\cite{Ariola97callbyneed}.

Philosophically, the by-need \lc has two implications:

\begin{enumerate}
\item First, its existence says that imperative assignment isn't truly
  needed to implement a lazy language. The calculus uses only
  one-at-a-time substitution and does not require any store-like
  structure. Instead, the by-need $\beta$ suggests that a variable 
  dereference is the resumption of a continuation of the function call, 
  an idea that Garcia et al.~\cite{Garcia09lazy} recently explored in detail
  by using delimited control operations to derive an abstract machine
  from the by-need calculus. Unlike traditional machines for lazy
  functional languages, Garcia et al.'s machine eliminates the need
  for a store by replacing heap manipulations with control (stack)
  manipulations.

\item Second, since by-need $\beta$ does not remove the application,
  the binding structure of programs---the association of a function parameter
  with its value---remains the same throughout a program's
  evaluation. This second connection is the subject of our paper. 
  This binding structure \emph{is} the control stack, and
  thus we have that in call-by-need, \emph{static} addresses can be
  resolved in the \emph{dynamic} control stack.
\end{enumerate}

Our key innovation is the CK+ machine, which refines the
abstract machine of Garcia et al. by making the observation that when
a variable reference is in focus, the location of the corresponding
binding context in the dynamic control stack can be determined by the
lexical index of the variable. Whereas Garcia et al.'s machine
linearly traverses their control stack to find a specific binding
context, our machine employs a different stack organization where indexing 
can be used instead of searching. Our machine organization also simplifies 
the hygiene checks used by Garia et al., mostly because it explicitly 
maintains Garcia et al.'s ``well-formedness'' condition on machine states, 
instead of leaving it as a side condition.

The paper starts with a summary of the by-need \lc and the abstract
textual machine induced by the standard reduction theorem. We then
show how to organize the machine's control stack so that when the
control string is a variable reference, the machine is able to use the
lexical address to compute the location of the variable's binding site
in the control stack.


\section{The Call-by-need \lc, the de Bruijn Version}
\label{debruijn}

The terms of the by-need \lc are those of the
\lc~\cite{Barendregt81lambda}, which we present using de Bruijn's
notation~\cite{deBruijn72lambda}, i.e., lexical addresses replace
variables:
$$M ::= n \mid \lamdb{M} \mid \app{M}{M}$$
where $n \in \mathbb{N}$. The set of values is just the set of abstractions:
$$V ::= \lamdb{M}$$

One of the fundamental ideas of call-by-need is to evaluate the
argument in an application only when it is ``needed,'' and when the
argument \emph{is} needed, to evaluate that argument only
once. Therefore, the by-need calculus cannot use the $\beta$ notion of
reduction because doing so may evaluate the argument when it is not needed, or may cause the argument to be evaluated multiple times. Instead, $\beta$ is replaced with the \emph{deref} notion of reduction:
\begin{align*}
\app{ \lamdbp{\inhole{E}{n}} }{ V }
  \needsp
\app{ \lamdbp{\inhole{E}{V}} }{ V }, \; \textrm{  $\lambda$ binds $n$ }
\tag*{\emph{deref}} 
\end{align*}
The \emph{deref} notion of reduction requires the argument in an
application to be a value and requires the body of the function to
have a special shape. This special shape captures the demand-driven
substitution of values for variables that is characteristic of
call-by-need. In the \emph{deref} notion of reduction, when a
variable is replaced with the value $V$, some renaming may still be
necessary to avoid capture of free variables in $V$, but for now, we
assume a variant of Barendregt's hygiene condition for de Bruijn
indices and leave all necessary renaming implicit.

Here is the set of evaluation contexts $E$:
$$E ::= \hole \mid \app{E}{M} \mid \app{ \lamdbp{E} }{ M } \mid \app{ \lamdbp{\inhole{E'}{n}} }{ E }$$
Like all contexts, an evaluation context is an expression with a hole
in the place of a subexpression. The first evaluation context is an
empty context that is just a hole. The second evaluation context
indicates that evaluation of applications proceeds in a
leftmost-outermost order. This is similar to how evaluation proceeds
in the by-name \lc~\cite{Plotkin75callbyname}. Unlike call-by-name,
however, call-by-need defers dealing with arguments until absolutely
necessary. It therefore demands evaluation within the body of a
let-like binding. The third evaluation context captures this
notion. This context allows the \emph{deref} notion of reduction to
search under applied $\lambda$s for variables to substitute. The
fourth evaluation context explains how the demand for a parameter's
value triggers and directs the evaluation of the function's
argument. In the fourth evaluation context, the visible $\lambda$ binds 
$n$ in $\lamdb{\inhole{E'}{n}}$. This means that there are $n$ additional 
$\lambda$ abstractions in $E'$ between $n$ and its binding $\lambda$.

To make this formal, let us define the function $\Delta : E \rightarrow \mathbb{N}$ as:

\begin{center}
\begin{tabular}{rclrcl}
$\Delta(\hole)$ & = & 0 & \hspace{1cm}
$\Delta(\app{ \lamdbp{\inhole{E'}{n}} }{E})$ & = & $\Delta(E)$ \\
$\Delta(\app{E}{M})$ & = & $\Delta(E)$ & 
$\Delta(\app{ \lamdbp{E} }{ M })$ & = & $\Delta(E) + 1$ \\
\end{tabular}
\end{center}

\noindent With $\Delta$, the side condition for the fourth evaluation
context is $n = \Delta(E')$.

Unlike $\beta$, \emph{deref} does not remove the argument from a
term when substitution is complete. Instead, a term $\app{ \lamdbp{M}
}{ N }$ is interpreted as a term $M$ and an environment where the
variable (index) bound by $\lambda$ is associated with $N$. Since
arguments are never removed from a by-need term, reduced terms are not
necessarily values. In the by-need \lc, reductions produce ``answers''
$a$ (this representation of answers is due to Garcia et
al.~\cite{Garcia09lazy}):
\begin{align*}
a &::= \inhole{A}{V} 
\tag*{answers} \\
A &::= \hole \mid \app{ \lamdbp{A} }{ M }
\tag*{answer contexts}
\end{align*}
Answer contexts $A$ are a strict subset of evaluation contexts $E$.

Since both the operator and the operand in an application reduce to answers, two additional notions of reduction are needed:
%
%
\begin{align*}
\app{ \app{ \lamdbp{\inhole{A}{V}} }{ M } }{ N }
  &\needsp
\app{ \lamdbp{\inhole{A}{\app{V}{N}}} }{ M }
\tag*{\emph{assoc-L}} \\
\app{ \lamdbp{\inhole{E}{n}} }{ \appp{ \lamdbp{\inhole{A}{V}} }{ M } }
  &\needsp
\app{ \lamdbp{ \inhole{A}{\app{ \lamdbp{\inhole{E}{n}} }{ V }}} }{ M }, \;\textrm{ if } \Delta(E) = n
\tag*{\emph{assoc-R}} 
\end{align*}

\label{debruijnrenaming}

As mentioned, some adjustments to de Bruijn indices are necessary when
performing substitution in \lc terms. For example, in a \emph{deref}
reduction, every free variable in the substituted $V$ must be
incremented by $\Delta(E)+1$. Otherwise, the indices representing free
variables in $V$ no longer count the number of $\lambda$s between
their occurrence and their respective binding $\lambda$s. Similar
adjustments are needed for the \emph{assoc-L} and \emph{assoc-R}
reductions, where subterms are also pulled under $\lambda$s.

Formally, define a function $\incFVsfnname$ that takes three inputs: a
term $M$, an integer $x$, and a variable (index) $m$, and increments
all free variables in $M$ by $x$, where a free variable is defined to
be an index $n$ such that $n \geq m$. In this paper, we use the
notation $\incFVs{M}{x}{m}$. Here is the formal definition of $\incFVsfnname$:

\begin{center}
\begin{tabular}{rclrcl}
$\incFVs{n}{x}{m}$ & = & $n+x, \textrm{ if } n \geq m$ & \hspace{1cm}
$\incFVs{\appp{M}{N}}{x}{m}$ & = & $\appp{ \incFVsp{M}{x}{m} }{ \incFVsp{N}{x}{m} }$ \\
$\incFVs{n}{x}{m}$ & = & $n, \hspace{6.3mm}\textrm{ if } n < m$ &
$\incFVs{\lamdb{M}}{x}{m}$ & = & $\lamdb{\incFVsp{M}{x}{m+1}}$ \\
\end{tabular}
\end{center}

Using the $\uparrow$ function for index adjustments, the notions of
reduction are:
\begin{align*}
\app{ \lamdbp{\inhole{E}{n}} }{ V }
  &\needsp
\app{ \lamdbp{\inhole{E}{ \incFVs{V}{\addone{\Delta(E)}}{0} }} }{ V }, \;\textrm{ if } \Delta(E) = n
\tag*{\emph{deref}} \\
\app{ \app{ \lamdbp{\inhole{A}{V}} }{ M } }{ N }
  &\needsp
\app{ \lamdbp{\inhole{A}{\app{V}{ \incFVsp{N}{\addone{\Delta(A)}}{0} }}} }{ M }
\tag*{\emph{assoc-L}} \\
\app{ \lamdbp{\inhole{E}{n}} }{ \appp{ \lamdbp{\inhole{A}{V}} }{ M } }
  &\needsp
\app{ \lamdbp{ \inhole{A}{\app{ \incFVsp{\lamdbp{\inhole{E}{n}}}{\addone{\Delta(A)}}{0} }{ V }}} }{ M }, \;\textrm{ if } \Delta(E) = n
\tag*{\emph{assoc-R}} 
\end{align*}
It is acceptable to apply the $\Delta$ function to $A$ because $A$ is
a subset of $E$.


\section{Standard Reduction Machine}

In order to derive an abstract machine from the by-need \lc, Ariola
and Felleisen prove a Curry-Feys-style Standardization Theorem. Roughly, the 
theorem states that a term $M$ reduces to a term $N$ in a canonical manner if 
$M$ reduces to $N$ in the by-need calculus.

The theorem thus determines a state machine for reducing programs to
answers. The initial state of the machine is the program, the
collection of states is all possible programs, and the final states
are answers. Transitions in the state machine are equivalent to
reductions in the calculus:
$$\inhole{E}{M} \onestepneed \inhole{E}{M'}, \textrm{ if } M \needsp M'$$
where $E$ represents the same evaluation contexts that are used to
define the demand-driven substitution of variables in the
\emph{deref} notion of reduction.

The machine is deterministic because all programs $M$ satisfy the
unique decomposition property. This means that $M$ is either an answer
or can be uniquely decomposed into an evaluation context and a redex. Hence, we
can use the state machine transitions to define an evaluator function:
$$
\texttt{eval}_{\textbf{need}}(M) = 
\begin{cases} 
  a,  & \mbox{if }M \multistepneed a \\
  \bot, & \mbox{if for all $M \multistepneed N$, $N \onestepneed L$ }  \\
\end{cases}
$$

\begin{lemma}
\label{totalfunlem1}
$\texttt{eval}_{\textbf{need}}$ is a total function.
\end{lemma}
\begin{proof}
The lemma follows from the standard reduction
theorem~\cite{Ariola97callbyneed}. 
\qed
\end{proof}


\section{The CK+ Machine}

A standard reduction machine specifies evaluation steps at a
high-level of abstraction. Specifically, at each evaluation step in
the machine, the entire program is partitioned into an evaluation
context and a redex. This repeated partitioning is inefficient because
the evaluation context at any given evaluation step tends to share a large 
common prefix with the evaluation context in the previous step.
To eliminate this inefficiency, Felleisen and Friedman propose the CK
machine~\cite[Chapter 6]{Felleisen09semantics}, an implementation for a standard reduction machine of a call-by-value language. Consider the following call-by-value evaluation:
\begin{align*}
&\appp{ \lamp{w}{w} }{ \apppu{ \lamp{x}{ \appp{ x }{ \appp{ \lamp{y}{y} }{ \lam{z}{z} } } } }{ \lam{x}{x} } }\\
\longmapsto_v \; &\appp{ \lamp{w}{w} }{ \appp{ \lamp{x}{x} }{ \apppu{ \lamp{y}{y} }{ \lam{z}{z} } } } \\
\longmapsto_v \; &\appp{ \lamp{w}{w} }{ \apppu{ \lamp{x}{x} }{ \lam{z}{z} } } \\
\longmapsto_v \; &\apppu{ \lamp{w}{w} }{ \lam{z}{z} } \\
\longmapsto_v \; &\lam{z}{z}
\end{align*}
In each step, the $\beta_v$ redex is underlined. The
evaluation contexts for the first and third term are the same, $\appp{
  \lamp{w}{w} }{ \hole }$, and it is contained in the evaluation
context for the second term, $\appp{ \lamp{w}{w} }{ \appp{ \lamp{x}{x}
  }{ \hole } }$. Although the evaluation contexts in the first three
terms have repeated parts, a standard reduction machine for the
call-by-value calculus must re-partition the program at each
evaluation step.

The CK machine improves upon the standard reduction machine for the by-value \lc by eliminating redundant search steps. While the standard reduction machine uses whole programs as machine states, a state in the CK machine is divided into separate subterm (C) and evaluation context (K) registers. More precisely, the C in the CK machine represents a control string, i.e., the subterm to be evaluated, and the K is a continuation, which is a data structure
that represents an evaluation context in an ``inside-out'' manner. The
original program can be reconstructed from a CK machine state by
``plugging'' the expression in the C subterm register into the context
represented by K. When the control string is a redex, the CK machine
can perform a reduction, just like the standard reduction
machine. Unlike the standard reduction machine though, the CK machine
still remembers the previous evaluation context in the context
register and can therefore resume the search for the next redex from
the contractum in C and the evaluation context in K.

\subsection{CK+ Machine States}

We introduce the CK+ machine, a variant of the CK machine, for the
by-need \lc. The CK+ machine is also a modification of the abstract
machine of Garcia et al.~\cite{Garcia09lazy}. The machine states for
the CK+ machine are specified in figure~\ref{fig:crksyntax}. The core
CK+ machine has three main registers, a control string ($C$), a
``renaming'' environment ($R$), and a continuation stack ($\Kstack$).

\begin{figure}[htbp]
\figuresepline
\begin{align*}
S,T ::= \;&\crk{ C }{ R }{ \Kstack } 
\tag*{machine states} \\
C ::= \;&M       
\tag*{control strings} \\
R ::= \;&(i, \; \ldots)
\tag*{renaming environments} \\
i \in \;&\mathbb{N}
\tag*{offsets} \\
\Kstack ::= \;&\makestack{k, \; K, \; \ldots}
\tag*{continuation stacks} \\
K ::= \;&\KbindRdb{M}{R}{k}
\tag*{complete stack frames} \\
k ::= \;&\Kmt \mid \KargR{M}{R}{k} \mid \Kopdb{\Kstack}{k} 
\tag*{partial stack frames} 
\end{align*}

\caption{CK+ machine states.}
\label{fig:crksyntax}

\figuresepline

\end{figure}

In figure~\ref{fig:crksyntax}, the $\ldots$ notation means ``zero or
more of the preceeding element'' and in the stack $\makestack{ k, \;
  K, \; \ldots }$, the partial stack frame $k$ is the top of the
stack.  The initial CK+ machine state is $\crk{ M }{ \emptyR }{
  \makestack{ \Kmt } }$, where $M$ is the given program, $\emptyR$ is
an empty renaming environment, and $\makestack{ \Kmt }$ is a stack
with just one element, an empty frame.

\subsection{Renaming Environment}

As mentioned in section~\ref{debruijnrenaming}, substitution requires
some form of renaming, which manifests itself as lexical address
adjustments when using a de Bruijn representation of terms. Instead of
adjusting addresses directly, the CK+ machine delays the adjustment by
keeping track of offsets for all free variables in the control string
in a separate \emph{renaming environment}.  The delayed renaming is
forced when a variable occurrence is evaluated, at which point the
offset is added to the variable before it is used to retrieve its value 
from the control stack.

Here we use lists for renaming environments and the
offset corresponding to variable $n$, denoted $R(n)$, is the $n$-th
element in $R$ (0-based). The : function is cons, and the function $\applyRtoM{R}{M}$ applies a
renaming environment $R$ to a term $M$, yielding a term like $M$
except with appropriately adjusted lexical addresses:
\begin{eqnarray*}
\applyRtoM{\emptyR}{M} & = & M 
\\
\applyRtoM{R}{n} & = & n + R(n) 
\\
\applyRtoM{R}{\lamdbp{M}}
  & = & 
\lamdb{\applyRtoMp{\consp{0}{R}}{M}} 
\\
\applyRtoM{R}{\appp{M}{N}}
  & = & 
\appp{\applyRtoMp{R}{M}}{\applyRtoMp{R}{N}}
\end{eqnarray*}

Because the CK+ machine uses renaming environments, the
$\incFVsfnname$ function from section~\ref{debruijnrenaming} is
replaced with an operation on $R$. When the machine needs to increment
all free variables in a term, it uses the $\addxfnname$ function to
increment all offsets in the renaming environment that accompanies the
term. The notation $\addx{R}{x}$ means that all offsets in renaming
environment $R$ are incremented by $x$. Thus, the use of indices
in place of variables enables hygiene maintenance through simple
incrementing and decrementing of the indices. As a result, we have eliminated the need to keep track of the ``active variables'' that are present in
Garcia et al.'s machine~\cite[Section 4.5]{Garcia09lazy}.

\subsection{Continuations and the Continuation Stack}

Like the CK machine, the CK+ machine represents evaluation contexts as
continuations. The $\hole$ context is represented by the $\Kmt$
continuation. An evaluation context $\inhole{E}{\appp{\hole}{N}}$ is
represented by a continuation $\KargR{M}{R}{k}$ where $k$ represents
$E$ and $\applyRtoMp{R}{M}=N$.  An evaluation context
$\inhole{E}{\app{\lamdbp{\hole}}{N}}$ is represented by a continuation
$\KbindRdb{M}{R}{k}$ where $k$ represents $E$ and
$\applyRtoMp{R}{M}=N$. Finally, the
$\inhole{E}{\app{\lamdbp{\inhole{E'}{n}}}{\hole}}$ context is
represented by an $\Kopdb{\Kstack}{k}$ continuation. The $E'$ under
the $\lambda$ in the evaluation context is represented by the nested
$\Kstack$ stack in the continuation and the $E$ surrounding the
evaluation context corresponds to the $k$ in the continuation. The
\texttt{op} continuation does not need to remember the $n$ variable in
the evaluation context because the variable can be derived from the
length of $\Kstack$.

The contents of the $\Kstack$ register represent the control stack of the program and we refer to an element of this stack as a frame. The key difference between the CK+ machine and Garcia et al.'s machine is in the organization of the frames of the stack. Instead of a flat list of frames like in Garcia et al.'s machine, our control stack frames are groups of nested continuations of a special shape. Thus we also call our control stack a ``continuation stack.'' We use two kinds of frames, partial and complete. The first frame in the continuation stack is always a partial one, while all others are complete. The
outermost continuation of a complete frame is a \texttt{bind} and all
other nested pieces of a complete frame are \texttt{op}, \texttt{arg},
or \texttt{mt}. Thus, not counting the first partial frame, there is
exactly one frame in the control stack for every \texttt{bind}
continuation in the program. As a result, the machine can use a
variable (lexical address) $n$ to find the \texttt{bind} corresponding
to that variable in the control stack.

\subsection{Maintaining the Continuation Stack}
\label{continuationstack}

Each frame of the control stack, with the exception of the top frame,
has the shape $\KbindRdb{M}{R}{k}$, where $k$ is a partial frame that
contains no additional \texttt{bind} frames. In order for the
continuation stack to maintain this invariant, CK+ machine transitions
must adhere to two conditions:

\begin{enumerate}
\item When a machine transition is executed, only the top partial
  frame of the stack is updated unless the instruction descends under
  a $\lambda$.
\item If a machine transition descends under a $\lambda$, the partial frame on 
top of the stack is completed and a new $\Kmt$ partial frame is pushed onto 
the stack.
\end{enumerate}

Essentially, the top frame in the stack ``accumulates context'' until
a $\lambda$ is encountered, at which time the top partial frame
becomes a complete frame. Maintaining evaluation contexts for the
program in this way implies a major consequence for the CK+ machine:

\begin{quote}
when the control string is a variable $n$, then the binding for $n$ is $(n + R(n) + 1)$ stack frames away.
\end{quote}

\subsection{Relating Machine States to Terms}

\begin{figure}[htb]

\figuresepline

\begin{minipage}{5.5cm}
\[
\begin{array}{rcll}
\phifn{ \crk{ M }{ R }{ \Kstack } }
  & = & 
\inhole{\Kstack}{\applyRtoM{R}{M}} 
\end{array}
\]
\end{minipage}
\begin{minipage}{3cm}
\[
\begin{array}{rcl}
\inhole{\makestack{k, \; K, \; \ldots}}{M} 
  & = & 
\inhole{\ldots}{\inhole{K}{\inhole{k}{M}}} 
\\[2mm]
\inhole{\Kmt}{M} & = & M 
\\
\inhole{\KargR{N}{R}{k}}{M} & = & \inhole{k}{(\app{M}{\applyRtoMp{R}{N}})}
\\
\inhole{\Kopdb{\Kstack}{k}}{M}
  & = & 
\inhole{ k }{ \app{ \lamdbp{\inhole{\Kstack}{\len{\Kstack}-1}} }{ M } } 
\\
\inhole{\KbindRdb{N}{R}{k}}{M}
  & = & 
\inhole{k}{\app{ \lamdbp{M} }{ \applyRtoMp{R}{N} }}
\end{array}
\]
\end{minipage}

\caption{$\phifnname$ converts CK+ machine states to \lc terms.}
\label{fig:phi}

\figuresepline

\end{figure}

Figure~\ref{fig:phi} defines the $\phifnname$ function, which converts
machine states to $\lambda$-terms.  It uses the $\applyRtoM{R}{M}$
function to apply the renaming environment to the control string and
then uses a family of ``plug'' functions, dubbed $\inhole{\cdot}{\cdot}$, to plug the renamed control string into the hole of the context represented by the continuation component of the state.  Figure~\ref{fig:phi} also defines these plug functions, where $\inhole{K}{M}$ yields the term obtained by plugging $M$ into the context represented by $K$, and $\inhole{\Kstack}{M}$ yields the term when $M$ is plugged into the context represented by the continuation stack $\Kstack$.

\subsection{CK+ Machine State Transitions}
\label{crkrules}

Figure~\ref{fig:crkmachine} shows the first four state transitions for
the CK+ machine. The $\appendfnname$ notation indicates an ``append''
operation for the continuation stack. Since the purpose of the CK+
machine is to remember intermediate states in the search for a redex,
three of the first four rules are search rules. They shift pieces of
the control string to the $\Kstack$ register. For example, the
[shift-arg] transition shifts the argument of an application to the
$\Kstack$ register.

\begin{figure}[htb]

\figuresepline
\small
\begin{center}
\begin{tabular}{lr}

\multicolumn{2}{c}{$\longmapsto_{ck+}$} \\
\hline & \\


 & \multicolumn{1}{r}{[shift-arg]} \\

$\crk{ \appp{M}{N} }
     { R }
     { \makestack{ k, \; K, \; \ldots } }$ & 
$\crk{ M }
    { R }
    { \makestack{ \KargR{N}{R}{k}, \; K, \; \ldots } }$ \\

 & \\ 
 

 & \multicolumn{1}{r}{[descend-$\lambda$]} \\

$\crk{ \lamdb{M} }
     { R }
     { \makestack{ \KargR{N}{R'}{k}, \; K, \; \ldots } }$ & 
\hspace{5mm} $\crk{ M }
     { \cons{0}{R} }
     { \makestack{ \Kmt, \; \KbindRdb{N}{R'}{k}, \; K, \; \ldots } }$ \\

 & \\


 & \multicolumn{1}{r}{[lookup-arg]} \\

$\crk{ n }
     { R }
     { \append{ \Kstack }
              { \makestack{ \KbindRdb{N}{R'}{k}, \; K, \; \ldots } } }$ &
$\crk{ N }
     { R' }
     { \makestack{ \Kopdb{\Kstack}{k}, \; K, \; \ldots } }$ \\

where 
$\len{\Kstack} = n+R(n)+1$ & \\


 & \multicolumn{1}{r}{[resume]} \\

$\crk{ V }
     { R }
     { \makestack{ \Kopdb{\Kstack}{k}, \; K, \; \ldots } }$ &
$\crk{ V }
     { R' }
     { \append{ \Kstack }
              { \makestack{ \KbindRdb{V}{R}{k}, \; K, \; \ldots } } }$ \\

 & where 
  $R' = \addx{R}{\len{\Kstack}}$ \\

 & \\ \hline 

\end{tabular}
\end{center}

\normalsize

\caption{State transitions for the CK+ machine.}
\label{fig:crkmachine}

\figuresepline

\end{figure}

The [descend-$\lambda$] transition shifts a $\lambda$ binding to the
$\Kstack$ register. When the control string in the CK+ machine is a
$\lambda$ abstraction, and that $\lambda$ is the operator in an
application term---indicated by an \texttt{arg} frame on top of the
stack---the body of the $\lambda$ becomes the control string; the top
frame in the stack is updated to be a complete \texttt{bind} frame;
and a new partial $\Kmt$ frame is pushed onto the stack.

The [descend-$\lambda$] instruction also updates the renaming
environment which, as mentioned, is a list of numbers. 
There is one offset in the renaming environment
for each \texttt{bind} continuation in the control stack and the
offsets in the renaming environment appear in the same order as
their corresponding \texttt{bind} continuations. When the machine descends into a $\lambda$ expression,
a new \texttt{bind} continuation is added to the top of the control
stack so a new corresponding offset is also added to the front of the
renaming environment. Since offsets are only added to the renaming environment
when the machine goes under a $\lambda$, whenever a variable $n$ (a
lexical address) becomes the control string, its renaming offset is located at
the $n$-th position in the renaming environment. 
A renaming offset keeps track of the relative position of a \texttt{bind}
continuation since it was added to the control stack so a [descend-$\lambda$] 
instruction adds a 0 offset to the renaming environment. 

When the control string is a variable $n$, the binding for $n$ is
accessed from the continuation stack by accessing the $(n + R(n) +
1)$-th frame in the stack. The [lookup-arg] instruction moves the
argument that is bound to the variable into the control string
register. The $\texttt{op}$ frame on top of the stack is updated to
store all the frames inside the binding $\lambda$, in the same order
that they appear in the stack. Using this strategy, the machine can
``jump'' back to this context after it is done evaluating the
argument. For a term $\app{ \lamdbp{\inhole{E}{n}} }{ M }$, this is
equivalent to evaluating $M$ while saving $E$ and then returning to
the location of $n$ after the argument $M$ has been evaluated. Note
that the [lookup-arg] transition does not perform substitution. The
argument has been copied into the control string register, but it has
also been removed from the continuation stack register.

When the frame on top of the stack is an $\texttt{op}$, it means the
current control string is an argument in an application term. When
that argument is a value, then a redex has been found and the value
should be substituted for the variable that represents it. The
[resume] rule is the only rule in figure~\ref{fig:crkmachine} that
performs a reduction in the sense of the by-need calculus. It is the
implementation of the \emph{deref} notion of reduction from the
calculus. Specifically, the [resume] rule realizes this substitution
by restoring the frames in the \texttt{op} frame back into the
continuation stack as well as copying the value into a new
\texttt{bind} frame. The result is nearly equivalent to the left hand
side of the [lookup-arg] rule except that the argument has been
evaluated and has been substituted for the variable.

Since the [resume] rule performs substitution, it must also update the
renaming environment. Hence, the distance between $V$ and its binding
frame is added to every offset in the renaming environment $R$, as
indicated by $\addx{R}{\len{\Kstack}}$. In other words, each offset in
the environment is being incremented by the number of \texttt{bind}
continuations that are added to the control stack.

In summary, the four rules of figure~\ref{fig:crkmachine} represent
intermediate partitions of the program into a subterm and an
evaluation context before a partitioning of the program into an
evaluation context and a \emph{deref} redex is found. As a result,
the CK+ machine does not need to repartition the entire program on
every machine step and is therefore more efficient than standard
reduction. To complete the machine now, we must make it deal with
answers.

\subsection{Dealing with Answers}
\label{dealingwithanswers}

The CK+ machine described so far has no mechanism to identify whether
a control string represents an answer. The by-need calculus, however,
assumes that it is possible to distinguish answers from terms on
several occasions, one of which is the completion of evaluation. To
efficiently identify answers, the CK+ machine uses a fourth ``answer''
register. The CK+ machine identifies answers by searching the
continuation stack for frames that are answer contexts. To distinguish
answer contexts from evaluation contexts, we characterize answer
contexts in~figure~\ref{fig:crkanswersyntax}.  A final machine state
has the form $\crksearch{ V }{ R }{ \emptystack }{ \Astack }$.

\begin{figure}[htbp]

\figuresepline
\begin{align*}
S,T ::= \;&\crk{C}{R}{\Kstack} \;\mid\; \crksearch{ V }{ R }{ \makestack{F, \; \ldots, \; K, \; \ldots} }{\Astack}
\tag*{machine states} \\
F ::= \;&\KbindRdb{M}{R}{\Kmt}
\tag*{answer (complete) frame} \\
\Astack ::= \;&\makestack{ \Kmt, \; F, \; \ldots }
\tag*{answer stacks} 
\end{align*}

\caption{CK+ machine answer states.}
\label{fig:crkanswersyntax}

\figuresepline

\end{figure}

When the control string is a value $V$ and $\Kmt$ is the topmost stack
frame, then some subterm in the program is an answer. In this
situation, the $\Kmt$ frame in the stack is followed by an arbitrary
number of $F$ frames. The machine searches for the answer by shifting
$\Kmt$ and $F$ frames from the continuation stack register to the
answer register. The machine continues searching until either a $K$
frame is seen or the end of the continuation stack is reached. If the
end of the continuation stack is reached, the entire term is an answer
and evaluation is complete.

\begin{figure}[htbp]

\figuresepline
\small
\begin{center}
\begin{tabular}{lr}

\multicolumn{2}{c}{$\longmapsto_{ck+}$} \\
\hline & \\


 & \multicolumn{1}{r}{[ans-search1]} \\

$\crk{ V }
     { R }
     { \makestack{ \Kmt, \; K, \; \ldots } }$ &
$\crksearch{ V }
           { R }
           { \makestack{ K, \; \ldots } }
           { \makestack{ \Kmt } }$ \\

 & \\


 & \multicolumn{1}{r}{[ans-search2]} \\

$\crksearch{ V }
           { R }
           { \makestack{ F', \; K, \; \ldots } }
           { \makestack{ \Kmt, \; F, \; \ldots } }$ &
$\crksearch{ V }
           { R }
           { \makestack{ K, \; \ldots } }
           { \makestack{ \Kmt, \; F, \; \ldots, \; F' } }$ \\

 & \\


%
%
%

 & \multicolumn{1}{r}{[assoc-L]} \\

\multicolumn{2}{l}{
  $\crksearch{ \lamdb{M'} }
             { R }
             { \makestack{ \KbindRdb{M}{R'}{\KargR{N}{R''}{k}}, \; K, \; \ldots } }
             { \makestack{ \Kmt, \; F, \; \ldots } }$ \hspace{3cm} } \\
& \\             
\multicolumn{2}{r}{             
$\crk{ M' }
     { \cons{0}{R} }
     { \makestack{ \Kmt, \; \KbindRdb{N}{R'''}{\Kmt}, \; F, \; \ldots, \; \KbindRdb{M}{R'}{k}, \; K, \; \ldots } }$ } \\

 & where $R''' = \addx{R''}{\len{ \makestack{ F, \; \ldots } } + 1}$ \\

 & \\ 
 

%
%

 & \multicolumn{1}{r}{[assoc-R]} \\
 
\multicolumn{2}{l}{
  $\crksearch{ V }
             { R }
             { \makestack{ \KbindRdb{M}{R'}{\Kopdb{\Kstack}{k}}, \;K, \ldots } }
             { \makestack{ \Kmt, \; F, \; \ldots } }$ } \\
& \\             
\multicolumn{2}{r}{             
$\crk{ V }
     { R'' }
     { \append{ \Kstack' }
              { \makestack{ \KbindRdb{V}{R}{\Kmt}, \; F, \ldots\, \; \KbindRdb{M}{R'}{k}, \; K, \; \ldots } } }$ } \\

\multicolumn{2}{r}{             
 where 
 $\Kstack' = \addx{\Kstack}{\len{ \makestack{ F, \; \ldots } } + 1}$, and $R'' = \addx{R}{\len{\Kstack'}}$ } \\
 
 & \\ \hline 

\end{tabular}
\end{center}

\normalsize

\caption{Transitions of the CK+ machine that handle answer terms.}
\label{fig:crkanswerrules}

\figuresepline

\end{figure}

The presence of a $K$ frame means an \emph{assoc-L} or an
\emph{assoc-R} redex has been found. In order to implement these
shifts, the CK+ machine requires four additional rules for handling
answers, as shown in figure~\ref{fig:crkanswerrules}. The
[ans-search1] rule shifts the $\Kmt$ frame to the answer register. The
[ans-search2] rule shifts $F$ frames to the answer register. The
[assoc-L] rule and the [assoc-R] rule roughly correspond to the
\emph{assoc-L} and \emph{assoc-R} notions of reduction in the
calculus, respectively. The rules are optimized versions of corresponding notions of reduction in the calculus because the transition after the reduction is always known. The [assoc-L] machine rule performs the equivalent of an \emph{assoc-L} reduction in the calculus, followed by a [descend-$\lambda$] machine transition. The [assoc-R] machine rule performs the equivalent of an \emph{assoc-R} reduction in the calculus, followed by a [resume] machine transition.

In figure~\ref{fig:crkanswerrules}, the function $\addxfnname$ has
been extended to a family of functions defined over renaming
environments, continuation stacks, and stack frames: $\addx{R}{x}$
increments every offset in the renaming environment $R$ by $x$ and the
function $\addx{\Kstack}{x}$ increments every offset in every renaming
environment in every frame in $\Kstack$ by $x$. The function $\len{
  \makestack{ F, \; \ldots } }$ returns the number of frames in
$\makestack{ F, \; \ldots }$. Maintaining the offsets in this manner
is equivalent to obeying Garcia et al.'s ``well-formedness'' condition
on machine states.

\subsection{Correctness}

Correctness means that the standard reduction machine and the CK+
machine define the same evaluator functions. Let us start with an
appropriate definition for the CK+ machine:
$$
\texttt{eval}_{ck+}(M) = 
\begin{cases} 
  a,  & \mbox{if }\crk{M}{\emptyR}{\makestack{\Kmt}} \multistepck \crksearch{V}{R}{\emptystack}{\Astack}\text, \\
       & \mbox{where } a = \phifn{ \crksearch{V}{R}{\emptystack}{\Astack} } \\
  \bot, & \mbox{if for all }\crk{M}{\emptyR}{\makestack{\Kmt}} \multistepck S, 
   S \onestepck T
\end{cases}
$$
Recall that the function $\phifnname$ converts CK+ machine states to
\lc terms (figure~\ref{fig:phi}). Here, $\phifnname$ has been extended to handle ``answer'' machine states:

$$\phifn{ \crksearch{ M }{ R }{ \Kstack }{ \Astack } }
   = 
\inhole{ \Kstack }{ \inhole{ \Astack }{ \applyRtoM{R}{M} } }$$

The desired theorem says that the two \texttt{eval} functions are equal.
\begin{theorem}\label{thmcorrect}
$\texttt{eval}_{\textbf{need}} = \texttt{eval}_{ck+}$.
\end{theorem}

To prove the theorem, we first establish some auxiliary lemmas on the
totality of $\texttt{eval}_{\textbf{ck+}}$ and the relation between
CK+ transitions and standard reduction transitions.

\begin{lemma}
\label{totalfunlem2}
$\texttt{eval}_{ck+}$ is a total function.

\end{lemma}
\begin{proof}
The lemma is proved via a subject reduction argument.
\qed
\end{proof}
%
%
%
%
The central lemma uses $\phifnname$ to relate CK+ machine transitions
to reductions.

\begin{lemma}
\label{mainlem1}

For all CK+ machine states $S$ and $T$, if $S \onestepck T$, then
either $\phifn{ S } \onestepneed \phifn{ T }$ or $\phifn{ S } =
\phifn{ T }$.

\end{lemma}
\begin{proof}
We proceed by case analysis on each machine transition, starting with
[resume]. Assume
\begin{align*}
&\crk{ V }
     { R }
     { \makestack{ \Kopdb{\Kstack}{k}, \; K, \; \ldots } }
   \onestepck \\
&\crk{ V }
     { \addx{R}{\len{ \Kstack }} }
     { \append{ \Kstack }
              { \makestack{ \KbindRdb{V}{R}{k}, \; K, \; \ldots } } } \enspace,
\end{align*}
then let
\begin{align*}
M_1 &= 
 \phifn{ \crk{ V }
             { R }
             { \makestack{ \Kopdb{\Kstack}{k}, \; K, \; \ldots } } } \\
 &= \inhole{ \makestack{K, \; \ldots} }
           { \inhole{ k }
                    { \app{ \lamdbp{\inhole{\Kstack}{\len{\Kstack}-1}} }
                          { \applyRtoMp{R}{V} } } } \\
M_2 &=
 \phifn{ \crk{ V }
             { \addx{R}{\len{ \Kstack }} }
             { \append{\Kstack}{ \makestack{ \KbindRdb{V}{R}{k}, \; K, \; \ldots }} } } \\
 &= \inhole{ \makestack{K, \; \ldots} }
           { \inhole{ k }
                    { \app{ \lamdbp{\inhole{\Kstack}{\applyRtoM{\addxp{R}{\len{\Kstack}}}{V}}} }
                          { \applyRtoMp{R}{V} } } } \enspace. 
\end{align*}
Since $M_1$ is a standard \emph{deref} redex, we have:
\begin{align*}
&\inhole{ \makestack{K, \; \ldots} }
       { \inhole{ k }
                { \app{ \lamdbp{\inhole{\Kstack}{\len{\Kstack}-1}} }
                      { \applyRtoMp{R}{V} } } } 
  \onestepneed \\
&\inhole{ \makestack{K, \; \ldots} }
       { \inhole{ k }
                { \app{ \lamdbp{\inhole{\Kstack}{ \incFVs{\applyRtoMp{R}{V}}{\len{\Kstack}}{0} }} }
                      { \applyRtoMp{R}{V} } } }
\end{align*}

\noindent To conclude that $M_1 \onestepneed M_2$ by the
\emph{deref} notion of reduction, we need to show:
$$\incFVs{\applyRtoMp{R}{V}}{\len{\Kstack}}{0} =
\applyRtoM{\addxp{R}{\len{\Kstack}}}{V}$$
Lemma~\ref{auxlem1} proves the general case for this
requirement. Therefore, we can conclude that $M_1 \onestepneed
M_2$. The proofs for [assoc-L] and [assoc-R] are similar.

As for the remaining instructions, they only shift subterms/contexts
back and forth between registers, so the proof is a straightforward
calculation.
\qed
\end{proof}

\begin{lemma}
\label{auxlem1}
%
$\forall R, R_1, R_2$, where $R = \append{R_1}{R_2}$ and $m = \len{R_1}$:

$$\incFVs{\applyRtoMp{R}{M}}{x}{m} =
  \applyRtoM{\appendp{R_1}{\addxp{R_2}{x}}}{M}$$

\end{lemma}
\begin{proof}
By structural induction on $M$.
\qed
\end{proof}

Using lemma~\ref{mainlem1}, the argument to prove our main theorem is
straightforward.

\begin{proof}[of Theorem~\ref{thmcorrect}]
We show $\texttt{eval}_{ck+}(M) = a \iff \texttt{eval}_{\textbf{need}}(M) = a$.

The left-to-right direction follows from the observation that for all
CK+ machine starting states $S$ and final machine states $S_\final$,
if $S \multistepck S_\final$, then \mbox{$M \multistepneed$} $a$,
where $\phifn{ S_\final } = a$. This is proved using
lemma~\ref{mainlem1} and induction on the length of the $\multistepck$
sequence.

The other direction is proved by contradiction. Assume
$\texttt{eval}_{\textbf{need}}(M) = a \neq \bot$ and $\texttt{eval}_{ck+}(M)
\neq a$. Since $\texttt{eval}_{ck+}$ is a total function, either:
\begin{enumerate}
\item $\crk{M}{\emptyR}{\makestack{\Kmt}} \multistepck S_\final$,
  where $\phifn{S_\final} \neq a$, or

\item the reduction of $\crk{M}{\emptyR}{\makestack{\Kmt}}$ diverges.
\end{enumerate}

It follows from the left-to-right direction of the theorem that, in
the first case, $\texttt{eval}_{\textbf{need}}(M) = \phifn{S_\final}
\neq a$, and in the second case, $\texttt{eval}_{\textbf{need}}(M) =
\bot$.
However, $\texttt{eval}_{\textbf{need}}(M) = a$ was assumed and
$\texttt{eval}_{\textbf{need}}$ is a total function, so a
contradiction has been reached in both cases.
Since none of the cases are possible, we conclude that if
$\texttt{eval}_{\textbf{need}}(M) = a$, then $\texttt{eval}_{ck+}(M) =
a$.
\qed
\end{proof}

\section{Stack Compacting}

Because the by-need $\lambda$-calculus does not substitute the argument of a function call for all occurrences of the parameter at once, applications are never removed. In the CK+ machine, arguments accumulate on the stack and remain there forever. For a finite machine, an ever-growing stack is a problem. In this section, we explain how to compact the stack. 

To implement a stack compaction algorithm in the CK+ machine, we
introduce a separate SC machine which removes all unused stack bindings from a CK+ machine state. Based on the SC
machine, the CK+ machine can be equipped with a non-deterministic [sc] transition:
\begin{align*}
\crk{ M }
    { R }
    { \Kstack }
  &\onestepcrk
\crk{ M }
    { R' }
    { \Kstack' }
\tag*{[sc]} \\
\textrm{where} 
\gc{ \FV{M}{R}{0} }
	 { (M,R) }
   { \Kstack }
   { \emptystack } 
  &\multistep_{sc}
\gc{ \mathcal{F} }
	 { (M,R') }
   { \emptystack }
   { \Kstack' }
\end{align*}

Figure~\ref{fig:scmachine} presents the SC machine. In this figure,
\FVfnname refers to a family of functions that extracts the set of
free variables from terms, stack frames, and continuation stacks. The
function \FVfnname takes a term $M$, a renaming environment $R$ and a
variable $m$, and extracts free variables from $M$, where a free
variable is defined to be all $n$ such that $n + R(n) \geq m$. The
function \FVfnname is similarly defined for stack frames and
continuation stacks. In addition, $\dec{\mathcal{F}}$ denotes the set obtained by decrementing every element in $\mathcal{F}$ by one. Finally, $\Kstack@k$ represents a frame merged appropriately into a continuation stack. For example, $\makestack{ k', \; K, \; \ldots, \KbindRdb{M}{R}{k''} }@k = \makestack{ k', \; K, \; \ldots, \;
  \KbindRdb{M}{R}{k''}@k }$, where $\KbindRdb{M}{R}{k''}@k =
\KbindRdb{M}{R}{k''@k}$, and so on, until finally $\Kmt@k = k$.

\begin{figure}[htb]
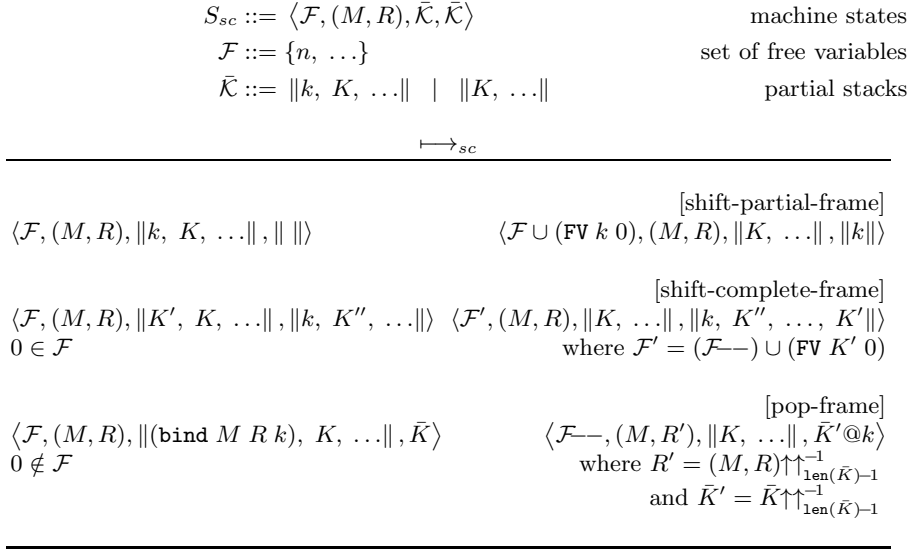


\figuresepline
\begin{align*}
S_{sc} ::= \;&\gc{ \mathcal{F} }{(M,R)}{ \Kcalstack }{ \Kcalstack } 
\tag*{machine states} \\
\mathcal{F} ::= \;&\{n, \; \ldots\}
\tag*{set of free variables} \\
\Kcalstack ::= \;&\makestack{k, \; K, \; \ldots} \;\;\mid\;\; \makestack{K, \; \ldots}
\tag*{partial stacks} 
\end{align*}
\begin{center}
\begin{tabular}{lr}

\multicolumn{2}{c}{$\longmapsto_{sc}$} \\

\hline & \\


 & \multicolumn{1}{r}{[shift-partial-frame]} \\

$\gc{ \mathcal{F} }
		{ (M,R) }
    { \makestack{k, \; K, \; \ldots} }
    { \emptystack }$ &
$\gc{ \mathcal{F} \cup \FVk{k}{0} }
		{ (M,R) }
    { \makestack{K, \; \ldots} }
    { \makestack{k} }$ \\

 & \\ 
 

 & \multicolumn{1}{r}{[shift-complete-frame]} \\

$\gc{ \mathcal{F} }
		{ (M,R) }
    { \makestack{K', \; K, \; \ldots} }
    { \makestack{k, \; K'', \; \ldots} }$ &
$\gc{ \mathcal{F}' }
		{ (M,R) }
    { \makestack{K, \; \ldots} }
    { \makestack{k, \; K'', \; \ldots, \; K'} }$ \\
$0 \in \mathcal{F}$ & where $\mathcal{F}' = \decp{\mathcal{F}} \cup \FVk{K'}{0}$ \\

 & \\


 & \multicolumn{1}{r}{[pop-frame]} \\

$\gc{ \mathcal{F} }
		{ (M,R) }
    { \makestack{\KbindRdb{M}{R}{k}, \; K, \; \ldots} }
    { \Kstack }$ &
$\gc{ \dec{\mathcal{F}} }
		{ (M,R') }
    { \makestack{K, \; \ldots} }
    { \Kstack'@k }$ \\
$0 \notin \mathcal{F}$ & 
 where $R' = \adjOffsetR{ M }{ R }{ -\!1 }{ \subone{\len{\Kstack}}\; }$\\
& and
 $\Kstack' = \adjOffsetK{ \Kstack }{ -\!1 }{ \subone{\len{\Kstack}}\; }$ \\

 & \\ \hline 
 
\end{tabular}
\end{center}

\caption{The SC machine.}
\label{fig:scmachine}

\figuresepline

\end{figure}

Also in figure~\ref{fig:scmachine}, $\adjOffsetfnname$ denotes a
family of functions that adjusts the offsets in renaming environments
to account for the fact that a $\lambda$ has been removed from the
term. If a variable $n$ refers to a \texttt{bind} stack frame that is
deeper in the stack than the frame that is removed, then the offset
for that variable needs to be decremented by one. A variable $n$
refers to a \texttt{bind} that is deeper than the removed frame if
$n+R(n)$ is greater than the depth of the removed frame. The
$\adjOffsetfnname$ function can be applied to renaming environments
directly or to continuation stacks or stack frames that contain
renaming environments. We use the notation
$\adjOffsetR{M}{R}{x}{\ell}$ to mean that the offsets in $R$ are
incremented by $x$ for all variables $n$ in $M$ where $n + R(n) >
\ell$. The result of $\adjOffsetR{M}{R}{x}{\ell}$ is a new renaming
environment with the adjusted offsets. The notation
$\adjOffsetK{\Kstack}{x}{\ell}$ means that the offsets for all $M$ and
$R$ pairs in the continuation stack $\Kstack$ are
adjusted. $\adjOffsetK{\Kstack}{x}{\ell}$ evaluates to a new
continuation stack that contains the adjusted renaming environments.

\section{Related Work and Conclusion}

The call-by-need calculus is due to Ariola et
al.~\cite{Ariola97callbyneed,Ariola95callbyneed,Maraist98callbyneed}. Garcia
et al.~\cite{Garcia09lazy} derive an abstract machine for Ariola and
Felleisen's calculus and, in the process, uncover a correspondence
between the by-need calculus and delimited control operations. Danvy
et al.~\cite{Danvy2010Defunctionalized} derive a machine similar to Garcia et al. by applying ``off-the-shelf'' transformations to the by-need calculus.

Our paper has focused on the binding structure of call-by-need
programs implied by Ariola and Felleisen's calculus. We have presented
the CK+ machine, which restructures the control stack of Garcia et
al.'s machine, and we have shown that lexical addresses can be used to
directly access binding sites for variables in this dynamic control stack, a 
first in the history of programming languages. The use of lexical addresses 
has also simplified hygiene maintenance by eliminating the need for the set of 
``active variables'' that is present in Garcia et al.'s machine states. In 
addition, we show how using indices in place of variables allows for simple 
maintenance of Garcia et al.'s ``well-formed'' machine states. Finally, we 
have presented a stack compaction algorithm, which is used in the CK+ machine 
to prevent stack overflow.
The compaction algorithm used in this paper is a restriction of the
more general garbage collection notion of reduction of Felleisen and
Hieb~\cite{Felleisen92Syntactic} and is also reminiscent of Kelsey's 
work~\cite{Kelsey93TailRecursive}.

\subsubsection*{Acknowledgments.} Thanks to the anonymous reviewers for their 
feedback and to Daniel Brown for inspiring discussions.

\bibliographystyle{splncs}
\bibliography{cbneed-tfp2010}

\begin{thebibliography}{10}

\bibitem{Plotkin75callbyname}
Plotkin, G.D.:
\newblock Call-by-name, call-by-value and the $\lambda$-calculus.
\newblock Theoretical Computer Science \textbf{1} (1975)  125--159

\bibitem{Ariola97callbyneed}
Ariola, Z.M., Felleisen, M.:
\newblock The call-by-need lambda calculus.
\newblock Journal of Functional Programming \textbf{7} (1997)  265--301

\bibitem{Garcia09lazy}
Garcia, R., Lumsdaine, A., Sabry, A.:
\newblock Lazy evaluation and delimited control.
\newblock In: Proceedings of the 36th Annual Symposium on Principles of
  Programming Languages, ACM (2009)  153--164

\bibitem{Barendregt81lambda}
Barendregt, H.P.:
\newblock The Lambda Calculus: Its Syntax and Semantics.
\newblock North Holland (1981)

\bibitem{deBruijn72lambda}
De~Bruijn, N.G.:
\newblock Lambda calculus notation with nameless dummies, a tool for automatic
  formula manipulation, with application to the {C}hurch-{R}osser theorem.
\newblock Indagationes Mathematicae (1972)  381--392

\bibitem{Felleisen09semantics}
Felleisen, M., Findler, R.B., Flatt, M.:
\newblock Semantics Engineering with PLT Redex.
\newblock MIT Press (2009)

\bibitem{Ariola95callbyneed}
Ariola, Z.M., Felleisen, M., Maraist, J., Odersky, M., Wadler, P.:
\newblock The call-by-need lambda calculus.
\newblock In: Proceedings of the 22nd Annual Symposium on Principles on
  Programming Languages. (1995)  233--246

\bibitem{Maraist98callbyneed}
Maraist, J., Odersky, M., Wadler, P.:
\newblock The call-by-need lambda calculus.
\newblock Journal of Functional Programming \textbf{8} (1998)  275--317

\bibitem{Danvy2010Defunctionalized}
Danvy, O., Millikin, K., Munk, J., Zerny, I.:
\newblock Defunctionalized interpreters for call-by-need evaluation.
\newblock In Blume, M., Vidal, G., eds.: 10th International Symposium on
  Functional and Logic Programming. Lecture Notes in Computer Science, Springer
  (2010)

\bibitem{Felleisen92Syntactic}
Felleisen, M., Hieb, R.:
\newblock The revised report on the syntactic theories of sequential control
  and state.
\newblock Theoretical Computer Science \textbf{103} (1992)  235--271

\bibitem{Kelsey93TailRecursive}
Kelsey, R.:
\newblock Tail-recursive stack disciplines for an interpreter.
\newblock Technical Report NU-CCS-93-03, Northeastern University (1993)

\end{thebibliography}

\end{document}